%% file: cdc18.tex
\documentclass[letterpaper, 10 pt, conference]{ieeeconf}  
\usepackage{amssymb}
\usepackage{amsmath}
\usepackage{graphicx}
\usepackage{comment}
\usepackage{booktabs}
\usepackage{breqn}
\usepackage{bm}
\usepackage{graphicx}
\usepackage{epstopdf}
\usepackage{amssymb}
\usepackage{wasysym}
\usepackage{float}
\usepackage{color}
\usepackage{dsfont}
\usepackage{algorithm,algorithmicx}
\usepackage{algpseudocode}
\usepackage{epstopdf}
\usepackage{tabularx}
\usepackage{cite}

\newtheorem{proposition}{Proposition}

\newtheorem{definition}{Definition}

\newtheorem{remark}{Remark}
\newtheorem{assumption}{Assumption}
\newcommand{\eod}{\ensuremath{\hfill\Box}}
\IEEEoverridecommandlockouts                              
\title{\LARGE \bf
	Resilient Distributed Energy Management for Systems of Interconnected Microgrids
}

\author{Wicak Ananduta, Jos\'{e} Mar\'{i}a Maestre,~\IEEEmembership{Member,~IEEE}, Carlos Ocampo-Martinez,~\IEEEmembership{Senior Member,~IEEE}, and \\Hideaki Ishii,~\IEEEmembership{Senior Member,~IEEE} 
	\thanks{W. Ananduta and C. Ocampo-Martinez are with the Automatic Control Department, Universitat Polit\`{e}cnica de Catalunya, Institut de Rob\`{o}tica i Inform\`{a}tica Industrial
		(CSIC-UPC), Barcelona, Spain
	       (emails: {\tt\small \{wananduta, cocampo\}@iri.upc.edu}).}%
    \thanks{J. M. Maestre and H. Ishii are with Department of Computer Science, Tokyo Institute of Technology, Yokohama, Japan (emails: {\tt \small  pepemaestre@us.es, ishii@c.titech.ac.jp}).}
    \thanks{J. M. Maestre is also with Department of System and Automation
	Engineering, University of Seville, Seville,
	Spain.}   
	\thanks{This work has received funding from the European Union's Horizon 2020 research and innovation programme under the Marie Sk\l{}odowska-Curie grant agreement No 675318 (INCITE). Financial support by the Spanish MINECO project DPI2017-86918-R and the Japanese Society for the Promotion of Science (scholarship PE16048) is also gratefully acknowledged.}
}

\begin{document}

\maketitle
\thispagestyle{empty}
\pagestyle{empty}

\begin{abstract}
	 In this paper, distributed energy management of interconnected microgrids, which is stated as a dynamic economic dispatch problem, is studied. Since the distributed approach requires cooperation of all local controllers, when some of them do not comply with the distributed algorithm that is applied to the system, the performance of the system might be compromised. Specifically, it is considered that adversarial agents (microgrids with their controllers) might implement control inputs that are different than the ones obtained from the distributed algorithm. By performing such  behavior, these agents might have better performance at the expense of deteriorating the performance of the regular agents. This paper proposes a methodology to deal with this type of adversarial agents such that we can still guarantee that the regular agents can still obtain feasible, though suboptimal, control inputs in the presence of adversarial behaviors. The methodology consists of two steps: (i) the robustification of the underlying optimization problem and (ii) the identification of adversarial agents, which uses hypothesis testing with Bayesian inference and requires to solve a local mixed-integer optimization problem. Furthermore, the proposed methodology also prevents the regular agents to be affected by the adversaries once the adversarial agents are identified. In addition, we also provide a sub-optimality certificate of the proposed methodology.
\end{abstract}
\section{Introduction}
\input{intro}

\section{Problem Formulation \& Distributed Approach}
\label{sec:prob_form}
\input{prob_form}

\section{Proposed Approach}
\label{sec:prop_ap}
In this section, the problem is reformulated such that the regular agents are robust against attacks and propose a methodology to identify the adversarial neighbors and to prevent an attack from them once they are identified.

\subsection{Robustification Against Attacks}
\label{sec:rob_at}
\input{rob_at}
\subsection{Attack Identification and Mitigation}
\label{sec:at_det}
\input{at_det}

\section{Case Study}
\label{sec:sim}
\input{case_st}
\section{Conclusion and Future Work}
\label{sec:concl}
\input{concl}

\bibliographystyle{ieeetran}
\bibliography{ref_cdc18}

\end{document}

%% file: intro.tex
In order to face the increasing penetration of distributed generation units, either dispatchable or non-dispatchable ones, and energy storages, such as batteries, supercapacitors, and fuel cells, in electrical networks, distributed approaches for energy management system currently gain a lot of attention, e.g., as discussed in \cite{pantoja2011,hans2018,baker2016,larsen2014}. The advantages of employing a distributed approach for this task include avoiding significant increase of information, communication, and modeling resources used for a centralized dispatch as well as distributing high computational burden \cite{pantoja2011}. 

In a distributed scheme, a distribution electrical network can be viewed as a system of interconnected microgrids \cite{pantoja2011,arefifar2012}, each of which is a controllable entity that has its own local controller. Therefore, the economic dispatch problem of the network must be decomposed and assigned to the local controllers. A distributed optimization approach can then be formulated and applied to solve the problem. In this regard, Model Predictive Control (MPC) strategy, with receding horizon principle, is suitable, particularly when the dynamics of the storages are considered, since the decisions/control inputs are always updated at each sampling time according to the measurement of the states.  Distributed MPC (DMPC) methods that have been proposed to solve economic dispatch problems include those that are based on dual decomposition \cite{larsen2014}, alternating direction method of multipliers (ADMM) \cite{hans2018}, optimality condition decomposition (OCD) \cite{baker2016} and population dynamics \cite{quijano2017}. These approaches are suitable since they are able to obtain an optimal solution given that the related optimization problem is convex. 

Two important features in such distributed approaches are the necessity to share information among the agents (in this case the microgrids) and the cooperation of the agents to apply the algorithm and to comply with the decisions obtained from the distributed algorithm. In this work, we deal with the problem of agent compliance, in which some of the agents do not always implement the decision obtained from the distributed algorithm. Instead, they may implement a different decision that is more beneficial for them but compromise the performance of the other agents and hence the entire system.  

Agents with such adversarial behaviors are identified in \cite{velarde2017} as \textit{liar agents}  or in \cite{sharma2017} as \textit{misbehaving agents}. The authors of \cite{velarde2017} propose a secure dual-decomposition-based DMPC, in which the agents that provide extreme control input values are monitored and disregarded, to deal with this issue. Furthermore, \cite{sharma2017} addresses a cyber-attack problem of a consensus-based distributed control scheme for distributed energy storage systems. The proposed approach in \cite{sharma2017} includes a fuzzy-logic-based detection and a consensus based leader-follower distributed control scheme.  Related to the cyber-security issue of cyber physical systems, in particular power systems, the work of \cite{pasqualetti2013} provides a mathematical framework for attack detection and monitoring. In addition, \cite{leblanc2013,dibaji2017,feng2017} and some of their references also discuss consensus problems in which some of the agents perform adversarial behavior to prevent convergence.

The contributions of this paper is as follows. We study the impact of an adversarial behavior in the distributed energy management system that is based on a DMPC scheme and propose to actively use the storage system and the possibility to establish/disestablish connections between agents to deal with this behavior. To this end, we propose an approach that consists of two main steps. The first step is the robustification of the economic dispatch problem. By considering the robust reformulation, we ensure that the regular agents always obtain a solution that satisfies all the constraints defined in the economic dispatch problem even though there are some agents that do not comply with the decisions. In the second step, we propose an active strategy to identify the adversarial agents that is based on hypothesis testing using Bayesian inference (e.g., \cite{hoff2009}). In this method, each regular agent must solve a local mixed-integer problem to decide the connections with its neighbors at each time instant. By actively connecting/disconnecting with neighbors, regular agents can then assess their hypothesis. Additionally, we also provide a decentralized sub-optimality certificate of our proposed approach. 

Differently from \cite{sharma2017}, we consider a DMPC scheme to act as an energy management. Thus, our work is more related to \cite{velarde2017} than the approaches discussed in \cite{sharma2017,leblanc2013,dibaji2017,feng2017}. However, the methodology that we propose in this paper is different than that proposed in \cite{velarde2017}, in a way that it is more specific for the aforementioned problem and particularly for power systems. Moreover, unlike \cite{velarde2017}, our approach can deal with more than one adversarial agent in a network. 

This paper is structured as follows. In Section \ref{sec:prob_form}, the dynamic economic dispatch problem of interconnected microgrids is formulated. Moreover, a distributed approach that is based on dual decomposition and the adversary model are presented. In Section \ref{sec:prop_ap}, the approach to deal with the adversarial behavior is proposed. Section \ref{sec:sim} provides the numerical simulations and Section \ref{sec:concl} concludes the paper.

\paragraph*{Notations}
The set of real numbers and integers are denoted by $\mathbb{R}$ and $\mathbb{Z}$, respectively. Moreover, $\mathbb{R}_{\geq a}$ denotes all real numbers in the set \{$b: b\geq a, \ b,a \in \mathbb{R}$\} and $\mathbb{Z}_{\geq a}$ denotes all integers in the set \{$b: b\geq a, \ b,a \in \mathbb{Z}$\}. A similar definition can be used for the strict inequality case. For column vectors $v_i$ with $i \in \mathcal{L}=\{l_1,\dots,l_{|\mathcal{L}|}\}$, the operator  $[v_i^{\top}]^{\top}_{i \in \mathcal{L}}$ denotes the column-wise concatenation, i.e., $[v_i^{\top}]^{\top}_{i \in \mathcal{L}} = [v_{l_1}^{\top},\cdots,v_{l_{|\mathcal{L}|}}^{\top}]^{\top}$. The vector $\mathds{1}_n$ denotes $[1 \ 1 \ \cdots \ 1]^{\top} \in \mathbb{R}^{n}$. The set cardinality and Euclidean norm are denoted by $|\cdot|$ and $\|\cdot\|_2$. Furthermore, $\mathbb{P}(\cdot)$ denotes the probability measure. Finally, discrete-time instants are denoted by the subscript $k$.

%% file: prob_form.tex
In this section, the dynamic economic dispatch is formulated as an MPC problem. Afterward, a DMPC strategy based on a distributed optimization approach is formulated for this problem. Finally, the adversaries are defined.

\subsection{Dynamic Economic Dispatch Problem}

Consider a network of interconnected microgrids, which can be represented as an undirected graph $\mathcal{S}=(\mathcal{N},\mathcal{E})$, where $\mathcal{N}=\{1,2,\dots,|\mathcal{N}|\}$ denotes the set of microgrids and $\mathcal{E} \subseteq \mathcal{N} \times \mathcal{N}$ denotes the set of physical links among the microgrids. In this regard, the link $(i,j) \in \mathcal{E}$ implies that it is possible to exchange energy between microgrids $i$ and $j$. Furthermore, denote the set of neighbors of microgrid $i$ by $\mathcal{N}_i$, i.e., $\mathcal{N}_i =\{j:(i,j)\in \mathcal{E} \}$. Each microgrid $i \in \mathcal{N}$ consists of an aggregated local load, denoted by $p^{\mathrm{d}}_{i,k} \in \mathbb{R}_{\geq 0}$, a set of dispatchable distributed generators, denoted by $\mathcal{G}_i$, and a storage system from which electrical energy can be stored and retrieved. Each microgrid can also obtain power by buying it from the main grid. In this economic dispatch problem, optimal power generation of the generators and storage usage are sought by considering their economical costs such that the loads are satisfied.   
Additionally, $p^{\mathrm{d}}_{i,k}$ is assumed to be bounded as follows:
\begin{equation}
|p^{\mathrm{d}}_{i,k} - \hat{p}^{\mathrm{d}}_{i,k}| \leq d^{\mathrm{max}}_{i}, \label{eq:pd_bound}
\end{equation}
where $\hat{p}^{\mathrm{d}}_{i,k},d^{\mathrm{max}}_{i} \in \mathbb{R}_{\geq 0}$ denote the forecast and the upper bound, respectively, which are assumed to be known a priori. Note that the forecast and bound can be obtained from historical data.
 
The power balance equations that must be satisfied by each microgrid $i \in \mathcal{N}$ at each time instant $k \in  \mathbb{Z}_{\geq 0}$ are as follows \cite{hans2018,baker2016}:
\begin{align}
\hat{p}^{\mathrm{d}}_{i,k} - p^{\mathrm{G}}_{i,k} - p^{\mathrm{st}}_{i,k} - p^{\mathrm{im}}_{i,k} - \sum_{j \in \mathcal{N}_i} p^{\mathrm{t}}_{ji,k} = 0,  \label{eq:pow_bal1}\\
p^{\mathrm{t}}_{ij,k} + p^{\mathrm{t}}_{ji,k} = 0, \quad \forall j \in \mathcal{N}_i, \label{eq:coup_cons}
\end{align} 
where $p^{\mathrm{G}}_{i,k} = \sum_{m \in \mathcal{G}_i}p^{\mathrm{g}}_{m,k} \in \mathbb{R}_{\geq 0}$ denotes the total power generation in microgrid $i$, with $p^{\mathrm{g}}_{m,k}$ being the power generation of distributed generator $m$; $p^{\mathrm{st}}_{i,k} \in \mathbb{R}$ denotes the power delivered by or to the storage; $p^{\mathrm{im}}_{i,k} \in \mathbb{R}_{\geq 0}$ denotes the imported power from the main grid; and ${p}^{\mathrm{t}}_{ji,k} \in \mathbb{R}$, for all $j \in \mathcal{N}_i$, denote the power flows between microgrids $i$ and $j$ and can be regarded as a coupled variable. Note that \eqref{eq:pow_bal1} resembles the DC approximation of the power flow equation, in which ${p}^{\mathrm{t}}_{ji,k}$ is a function of the voltage angles. Furthermore, \eqref{eq:coup_cons} ensures that there is an agreement between two neighboring microgrids in terms of the power exchanged between them. 

The dynamics of the storage system, for each $i \in \mathcal{N}$, is represented as follows:
\begin{equation}
x_{i,k+1} = a_i x_{i,k} + b_i p^{\mathrm{st}}_{i,k}, \label{eq:dyn_bat}
\end{equation}
where $x_{i,k}$ denotes the state-of-charge (SoC) of storage $i$, $a_i \in (0,1]$ denotes the the efficiency of the storage and $b_i~=~-\frac{T_{\mathrm{s}}}{e_{\mathrm{cap},i}}$, where $T_{\mathrm{s}}$ and $e_{\mathrm{cap},i}$ denote the sampling time and the maximum capacity of the storage, respectively. 

Additionally, for each microgrid $i \in \mathcal{N}$, some local operational constraints are also considered as follows:
\begin{align}
x^{\mathrm{min}}_i &\leq x_{i,k} \leq x^{\mathrm{max}}_i,  \label{eq:cap_bat}\\
-p^{\mathrm{ch}}_i &\leq p^{\mathrm{st}}_{i,k} \leq p^{\mathrm{dh}}_i,\label{eq:ch_bat} \\
p^{\mathrm{G},\mathrm{min}}_i &\leq p^{\mathrm{G}}_{i,k} \leq p^{\mathrm{G},\mathrm{max}}_i, \label{eq:cap_pg}\\
& \quad \ p^{\mathrm{im}}_{i,k} \leq p^{\mathrm{im,max}}_{i} \label{eq:cap_im}\\ 
-p^{\mathrm{t},\mathrm{max}}_{ji} &\leq p^{\mathrm{t}}_{ji,k} \leq p^{\mathrm{t},\mathrm{max}}_{ji}, \quad \forall j \in\mathcal{N}_i, \label{eq:cap_t}
\end{align}
where $x^{\mathrm{min}}_i,x^{\mathrm{max}}_i  \in \mathbb{R}_{\geq 0}$  denote the minimum and the maximum SoC of the storage of microgrid $i$, respectively. Note that $0 \leq x^{\mathrm{min}}_i \leq x^{\mathrm{max}}_i \leq 1$.  
Moreover, $p^{\mathrm{ch}}_i \in \mathbb{R}_{\geq 0}$ and $p^{\mathrm{dh}}_i \in \mathbb{R}_{\geq 0}$ denote the maximum charging and discharging power of the storage. Furthermore, $p^{\mathrm{G},\mathrm{min}}_i,p^{\mathrm{G},\mathrm{max}}_i  \in \mathbb{R}_{\geq 0}$ denote the minimum and the maximum power generated by the distributed generators of microgrid $i$, respectively, $p^{\mathrm{im,max}}_{i}$ denotes the maximum imported power from the main grid, and  $p^{\mathrm{t},\mathrm{max}}_{ji}$ denotes the maximum energy that can be transferred between microgrid $i$ and $j$. Notice that \eqref{eq:cap_t} is symmetric and $p^{\mathrm{t},\mathrm{max}}_{ji} = p^{\mathrm{t},\mathrm{max}}_{ij}$, for all $(i,j) \in \mathcal{E}$.

Now,
denote the control input vector of microgrid $i$ by $\bm{u}_{i,k} = [ p^{\mathrm{st}}_{i,k} \ p^{\mathrm{G}}_{i,k}  \ p^{\mathrm{im}}_{i,k} \ \bm{u}^{\mathrm{c}\top}_{i,k} ]^{\top} \in \mathbb{R}^{3+|\mathcal{N}_i|}$, where $\bm{u}^{\mathrm{c}}_{i,k}~=~[p^{\mathrm{t}}_{ji,k}]^{\top}_{j \in \mathcal{N}_i}$ is the vector of coupled control input variables. 
We denote $h_p$ as the prediction horizon and consider the quadratic cost function
\begin{equation}
	J_{i,k} = \bm{u}_{i,k}^{\top}R_i\bm{u}_{i,k}, \label{eq:cost_func}
\end{equation}
where $R_i = \mathrm{diag}([ c^{\mathrm{st}}_{i} \ c^{{G}}_{i}  \ c^{\mathrm{im}}_{i} \ {c}^{\mathrm{t}}_{i}\mathds{1}^{\top}_{|\mathcal{N}_i|} ])> 0$, in which $c^{\mathrm{st}}_{i}, \ c^{{G}}_{i},  \ c^{\mathrm{im}}_{i},{c}^{\mathrm{t}}_{i} \in \mathbb{R}_{>0}$ denote the cost of storage operation, the cost of producing energy, the cost of buying energy from the main grid, and the cost of transferring energy to/from the neighbor due to losses \cite{hans2018}. 
Thus, the finite-time optimization problem that underlies an MPC strategy for the dynamic economic dispatch of this system can be written as
\begin{subequations}
	\begin{align}
	&\underset{\{\{\bm{u}_{i,\ell|k}\}_{i \in \mathcal{N}}\}_{\ell=k}^{k+h_p-1}}{\text{minimize}}  \sum_{i \in \mathcal{N}} \sum_{\ell =k}^{k+h_p-1}  J_{i,\ell}(\bm{u}_{i,\ell|k})\\
	 &\quad \quad \text{subject to }  
	  \quad \bm{F}_i{\bm{u}}_{i,\ell|k} \leq {\bm{f}}_{i,\ell}, \ \forall i \in \mathcal{N}, \label{eq:mg_loc_cons}\\
	  &\quad \quad \quad{\bm{u}}^{\mathrm{c}}_{i,\ell|k} + \sum_{j \in \mathcal{N}_i} \bm{G}_{ij}{\bm{u}}^{\mathrm{c}}_{j,\ell|k} = \mathbf{0}, \ \forall i \in \mathcal{N}, \label{eq:mg_coup_cons}
	\end{align}
	\label{eq:MPC_mg}%
\end{subequations}
for all $\ell \in \{k,\dots,k+h_p-1\}$,
where the local constraints \eqref{eq:mg_loc_cons} that only include local control inputs are constructed from \eqref{eq:pow_bal1}, \eqref{eq:dyn_bat}-\eqref{eq:cap_t}, while the coupled constraints \eqref{eq:mg_coup_cons} are constructed from \eqref{eq:coup_cons}.
\begin{remark}
	{Without loss of generality, $p^{\mathrm{G}}_{i,k}$ is considered as one of the control input instead of $p^{\mathrm{g}}_{m,k}$, for all $m \in \mathcal{G}_i$, for simplicity of the exposition. Considering $p^{\mathrm{g}}_{m,k}$, for all $m \in \mathcal{G}_i$, in $\bm{u}_i$ is also straightforward and only increases the dimension of $\bm{u}_i$.} \eod
\end{remark}
\begin{remark}
	In the matrix $R_i$, the weight/cost of exchanging energy, $c_i^{\mathrm{t}}$, is considered to be smaller than the other weights. \eod
\end{remark}

\begin{remark}
	 Problem \eqref{eq:MPC_mg} considers the load forecast, which does not always match with the actual load. Therefore, the proposed robust reformulation in Section \ref{sec:rob_at} takes into account the fact that $p_{i,k}^{\mathrm{d}}$, for all $i \in \mathcal{N}$, are bounded, as expressed in \eqref{eq:pd_bound}. \eod
\end{remark}

Problem \eqref{eq:MPC_mg} is convex since the inequality constraints form a polyhedron, the coupled equality constraints are affine, and the cost function \eqref{eq:cost_func} is strictly convex. 
Furthermore, the following assumption is considered.
\begin{assumption}
	\label{as:feas_sol}
	{For Problem \eqref{eq:MPC_mg}, there exists a nonempty set of feasible solutions and it includes a subset in which $p_{ij,k}^{\mathrm{t}} =p_{ji,k}^{\mathrm{t}} = 0$, for any $(i,j) \in \mathcal{E}$ and $k \in \mathbb{Z}_{\geq 0}$.} \eod
\end{assumption}

 Note that $p_{ij,k}^{\mathrm{t}} =p_{ji,k}^{\mathrm{t}} = 0$ implies that there is no power exchanged between microgrids $i$ and $j$. Based on this assumption, it is considered that each microgrid is able to satisfy its load independently, e.g., in the island mode. However, it is more cost efficient if the microgrids exchange power among them when they are connected.
 
\subsection{Distributed Energy Management based on Dual Decomposition }
In general, many distributed optimization algorithm can be applied as a DMPC strategy to solve Problem \eqref{eq:MPC_mg}. However, for the clarity of the explanation, a DMPC algorithm based on dual decomposition is considered in this paper. It is known that the solution obtained from a distributed algorithm based on dual decomposition converges to the optimal solution if the problem is convex with strictly convex cost function \cite{necoara2008}. In order to design the mentioned algorithm, the Lagrangian function associated to Problem \eqref{eq:MPC_mg} is derived and its dual problem \cite{boyd2010} is decomposed into smaller problems that are assigned to the agents (microgrids).  
The DMPC strategy based on dual decomposition is stated in Algorithm \ref{alg:dual_dec},  
where $\bm{\lambda}_{i,\ell} \in \mathbb{R}^{|\mathcal{N}_i|}$, for all $\ell \in \{k,\dots,k+h_p-1\}$ and all $i \in \mathcal{N}$, are the Lagrange multipliers associated to the coupled constraints \eqref{eq:mg_coup_cons}. In this algorithm, each agent should solve the local optimization problem in step \ref{eq:lo_dd} and update its Lagrange multipliers via the gradient-ascent method at each iteration. Finally, denote the optimal decisions obtained by the DMPC strategy for time $k$ by $\bm{u}_{i,k|k}^{\star}$, for all $i \in \mathcal{N}$. 
\begin{algorithm}
	\caption{DMPC algorithm based on dual decomposition, for each agent $i \in \mathcal{N}$ }
	\begin{algorithmic}[1]
		\State Set $r=1$, $\varepsilon \in \mathbb{R}_{>0}$, and initialize $\bm{\lambda}_{i,\ell}^{(r)}$
		\While{$\left| \left| \left[\begin{matrix}
			\bm{\psi}_{{i,k}}^{\top} & \cdots & \bm{\psi}_{{i,k+h_p-1}}^{\top}
			\end{matrix}\right] \right|\right|_2 > \varepsilon$}
		\State Receive $\bm{\lambda}_{j,\ell}^{(r)}$ for all $\ell \in \{k,\dots,k+h_p-1 \}$ from the neighbors, all $j \in \mathcal{N}_i$, and send $\bm{\lambda}_{i,\ell}^{(r)}$ for all $\ell \in \{k,\dots,k+h_p-1 \}$ to the neighbors
		\State \label{eq:lo_dd} Solve the local optimization problem:
		\begin{equation*}
		\begin{aligned}
		\underset{\{\bm{u}_{i,\ell|k}\}_{\ell=k}^{k+h_p-1}}{\text{minimize}} & \sum_{\ell =k}^{k+h_p-1} \left(  J_{i,\ell}(\bm{u}_{i,\ell|k}) +  \bm{y}_{i,\ell}^{\top}\bm{u}^{c}_{i,\ell|k}\right)\\
		\text{subject to} \quad  &  \text{\eqref{eq:mg_loc_cons}}, \quad \forall \ell \in \{k,\dots,k+h_p-1 \},
		\end{aligned}		
		\end{equation*}
		where $\bm{y}_{i,\ell}^{\top}= \bm{\lambda}^{(r)\top}_{i,\ell} + \sum_{j \in \mathcal{N}_i} \bm{\lambda}_{j,\ell}^{(r)\top}\bm{G}_{ji}$
		\State Receive the decision $\bm{u}^{\mathrm{c}}_{j,\ell|k}$ for all $\ell \in \{k,\dots,k+h_p-1 \}$ from the neighbors, all $j \in \mathcal{N}_i$, and send $\bm{u}^{\mathrm{c}}_{i,\ell|k}$ for all $\ell \in \{k,\cdots,k+h_p-1 \}$ to the neighbors
		\State Update $\bm{\lambda}_{i,\ell}$ for all $\ell \in \{k,\dots,k+h_p-1 \}$ as
		\begin{equation*}
		\bm{\lambda}_{i,\ell}^{(r+1)} = \bm{\lambda}_{i,\ell}^{(r)} + \gamma\bm{\psi}_{{i,\ell}},
		\end{equation*}
		where  $\bm{\psi}_{{i,\ell}}= \left(\bm{u}^{\mathrm{c}}_{i,\ell|k} + \sum_{j \in \mathcal{N}_i} \bm{G}_{ij}\bm{u}^{\mathrm{c}}_{j,\ell|k}\right)$ and $0~<~\gamma~<~1$
		\State $r \leftarrow r+1$
		\EndWhile
	\end{algorithmic}
	\label{alg:dual_dec}
\end{algorithm}

\subsection{Adversary Model}

The agents are classified as regular and adversarial agents based on the following definitions.
\begin{definition}
	{Agent $i$ belongs to the set of regular agents, denoted by $\mathcal{R}$, if it always implements its control input $\bm{u}_{i,k}$ according to the decision obtained from the DMPC strategy, i.e., $\bm{u}_{i,k}=\bm{u}_{i,k|k}^{\star}$, for all $k \geq 0$. Otherwise, agent $i$ belongs to the set of adversarial agents, denoted by $\mathcal{A}$.} \eod
\end{definition}

\begin{definition}\label{def:attack}
	{An attack is defined as the event at one time instant when at least one adversarial agent implements its control input that is different than the decision obtained from the DMPC strategy.} \eod
\end{definition}

We consider the $f$-local model of adversaries, which is stated in Definition \ref{def:f-local}.
\begin{definition}[{\hspace{-1pt}\cite{leblanc2013}}]
	{The set of adversarial agents is $f$-local if $|\mathcal{A} \cap \mathcal{N}_i| \leq f$, for $f \in \mathbb{Z}_{\geq 1}$ and all $i \in \mathcal{N}$.} \label{def:f-local}\eod
\end{definition}

In this paper, the case is restricted for $f=1$, as stated in the following Assumption \ref{assu:f=1}.
\begin{assumption}
	{Each agent has at most one adversarial neighbor.} \eod \label{assu:f=1}
\end{assumption}

\begin{assumption}
	{Regular agents do not have prior knowledge of the occurrence of the attacks, but they have an initial expectation on the probability of attacks, denoted by $P_{\mathrm{at}} \in (0,1]$.} \eod
\end{assumption}

The adversarial agents may try to gain advantage by implementing a different decision that benefits these agents. In the economic dispatch problem, the adversarial agents may get benefit if they decide to reduce the energy production and/or store more energy to their storages. Therefore, in order to meet their power balance equation, they ask their neighbors to provide the deficiency i.e., ${p}_{ij,k}^{\mathrm{t}} > p_{ij,k|k}^{\mathrm{t}\star}$, for $j \in \mathcal{A}$ and $i \in \mathcal{R}$, where $p_{ij,k|k}^{\mathrm{t}\star}$ denotes the decision obtained from the DMPC method.  Although it leads to a global suboptimal solution, the adversarial agents gain an advantage locally by performing this action. In other words, the adversarial agents are not willing to cooperate for their own interest. It is also possible that this behavior is observed due to a fault in the adversarial agents.

%% file: rob_at.tex
Regular agents might be affected negatively from the attacks of their adversarial neighbors. Due to the coupled constraints \eqref{eq:coup_cons}, regular agents must conform with the actions taken by their adversarial neighbors. For instance, if the adversarial neighbor $j \in \mathcal{A}$ requests more power than the agreed solution, then the regular microgrids $i \in  \mathcal{N}_j$ must adjust their decision (control inputs $\bm{u}_{i,k}$) in order to satisfy their power balance  \eqref{eq:pow_bal1}. In this regard, the existence of a storage unit at each microgrid could help to mitigate this issue without affecting the operation of the distributed generators. Additionally, uncertain loads might have similar effect to all microgrids and we consider that the deviation between the forecast and the actual load is compensated by the storage units.

In order to meet the power balance \eqref{eq:pow_bal1} when an attack occurs, more power from the storage ($p^{\mathrm{st}}_{i,k}$) is taken. However, it implies that the evolution of the SoC is different than the one that is predicted by the dynamic model \eqref{eq:dyn_bat}. Due to this circumstance, it may happen that the minimum limit of the storage capacity \eqref{eq:cap_bat} is violated. 

In order to ensure that there is no violation on the constraints, a formulation that robustifies Problem \eqref{eq:MPC_mg} against such attacks as well as the uncertainty of the load is proposed.  To this end, we consider the attack as disturbance, denoted by $w_{i,k}^{\mathrm{a}}$, and denote the load disturbance by $w_{i,k}^{\mathrm{d}}$. These disturbances affect the power balance \eqref{eq:pow_bal1} as follows:
\begin{equation}
	\hat{p}^{\mathrm{d}}_{i,k} - p^{\mathrm{G}}_{i,k} - {p}^{\mathrm{st}}_{i,k} - p^{\mathrm{im}}_{i,k} -w_{i,k}^{\mathrm{d}}- w_{i,k}^{\mathrm{a}}- \sum_{j \in \mathcal{N}_i} p^{\mathrm{t}}_{ji,k}  = 0.
\end{equation}
Although $w_{i,k}^{\mathrm{d}}$ and $w_{i,k}^{\mathrm{a}}$ are uncertain, they are bounded by \eqref{eq:pd_bound} and \eqref{eq:cap_t}, respectively. Therefore, agent $i \in \mathcal{R}$ might consider the worst case of the total disturbance, denoted by $w_{i,k}=w_{i,k}^{\mathrm{a}}+w_{i,k}^{\mathrm{d}}$, which is stated as follows:
\begin{align}
w_{i,k}^{\mathrm{max}} = \max_{j \in \mathcal{N}_i}\left(2p^{\mathrm{t},\mathrm{max}}_{ji}\right) +d_i^{\mathrm{max}}, \label{eq:w_max}
\end{align}
due to \eqref{eq:cap_t} and Assumption \ref{assu:f=1}. 
Since $w_{i,k}$ is compensated by the power delivered by/to the storage ${p}^{\mathrm{st}}_{i,k}$, the constraints related to ${p}^{\mathrm{st}}_{i,k}$, i.e., \eqref{eq:cap_bat} and \eqref{eq:ch_bat}, might be violated. 
 Therefore, these constraints are tightened to accommodate the worst case disturbance $w_{i,k}^{\mathrm{max}}$  as follows:
\begin{align}
x^{\mathrm{min}}_i- b_i w_{i,k}^{\mathrm{max}}&\leq a_i x_{i,\ell} +b_i {p}^{\mathrm{st}}_{i,\ell} \leq x^{\mathrm{max}}_i+ b_i w_{i,k}^{\mathrm{max}},  \label{eq:cap_bat_r}\\
-p^{\mathrm{ch}}_i + w_{i,k}^{\mathrm{max}} &\leq {p}^{\mathrm{st}}_{i,\ell} \leq p^{\mathrm{dh}}_i - w_{i,k}^{\mathrm{max}},\label{eq:ch_bat_r}
\end{align}
for all $\ell \in \{k,\dots,k+h_p-1\}$. Hence, the robust reformulation of Problem \eqref{eq:MPC_mg} is stated as follows:
\begin{subequations}
	\begin{align}
	&\underset{\{\{\bm{u}_{i,\ell|k}\}_{i \in \mathcal{N}}\}_{\ell=k}^{k+h_p-1}}{\text{minimize}}  \sum_{i \in \mathcal{N}} \sum_{\ell =k}^{k+h_p-1}  J_{i,\ell}(\bm{u}_{i,\ell|k})\\
	&\quad \quad \text{subject to }  
	 {\bm{F}}_i^{\mathrm{r}}{\bm{u}}_{i,\ell|k} \leq {{\bm{f}}}_{i,\ell}^{\mathrm{r}}, \ \forall i \in \mathcal{N}, \label{eq:mg_loc_cons_r}\\
	&\quad \quad \quad {\bm{u}}^{\mathrm{c}}_{i,\ell|k} + \sum_{j \in \mathcal{N}_i} \bm{G}_{ij}{\bm{u}}^{\mathrm{c}}_{j,\ell|k} = \mathbf{0}, \ \forall i \in \mathcal{N}, \label{eq:mg_coup_cons_r}
	\end{align}
	\label{eq:MPC_mg_r}%
\end{subequations}
for all $\ell \in \{k,\dots,k+h_p-1\}$,
where \eqref{eq:mg_loc_cons_r} with the appropriate ${\bm{F}}_i^{\mathrm{r}}$ and ${{\bm{f}}}_{i,\ell}^{\mathrm{r}}$ is defined according to \eqref{eq:pow_bal1}, \eqref{eq:dyn_bat}, \eqref{eq:cap_pg}-\eqref{eq:cap_t}, and \eqref{eq:w_max}-\eqref{eq:ch_bat_r}.

\begin{proposition}\label{prop:1}
Suppose that Assumption \ref{as:feas_sol} holds. Problem \eqref{eq:MPC_mg_r} has feasible solutions  if and only if 
	\begin{equation}
		w_{i,k}^{\mathrm{max}} \leq \min \left(\frac{1}{2}\left(p^{\mathrm{ch}}_i+p^{\mathrm{dh}}_i\right),-\frac{1}{2b_i}\left(x^{\mathrm{max}}_i-x^{\mathrm{min}}_i\right) \right).
		\label{eq:cond_prop1}
	\end{equation}
Furthermore, suppose that both Assumption \ref{assu:f=1} and \eqref{eq:cond_prop1} hold. Then, any feasible solution of Problem \eqref{eq:MPC_mg_r} does not violate operational constraints \eqref{eq:pow_bal1}-\eqref{eq:cap_t} even though an attack, which is defined in Definition \ref{def:attack}, occurs.	\eod
\end{proposition}
\begin{proof}
	The difference between Problems \eqref{eq:MPC_mg} and \eqref{eq:MPC_mg_r} is the fact that the tightened constraints \eqref{eq:cap_bat_r} and \eqref{eq:ch_bat_r} are considered in Problem \eqref{eq:MPC_mg_r}. Thus, a feasible region exists if and only if  $x^{\mathrm{min}}_i- b_i w_{i,k}^{\mathrm{max}}  \leq x^{\mathrm{max}}_i+ b_i w_{i,k}^{\mathrm{max}}$ and 
	$-p^{\mathrm{ch}}_i + w_{i,k}^{\mathrm{max}}  \leq p^{\mathrm{dh}}_i - w_{i,k}^{\mathrm{max}}$. The necessary and sufficient condition \eqref{eq:cond_prop1} is obtained from these two inequalities. Provided that a feasible solution of Problem \eqref{eq:MPC_mg_r} exists, the second claim follows from the formulation of Problem \eqref{eq:MPC_mg_r}. 	 
\end{proof}

If the condition of $w_{i,k}^{\mathrm{max}}$ stated in Proposition \ref{prop:1} is not satisfied, then $p^{\mathrm{G}}_{i,k}$ and/or $p^{\mathrm{im}}_{i,k}$ must also be involved in compensating $w_{i,k}$. In this regard, the constraints related to $p^{\mathrm{G}}_{i,k}$ and $p^{\mathrm{im}}_{i,k}$ must be tightened with similar procedure as that previously explained. For the remaining of the paper, suppose that 
the next assumption holds.

\begin{assumption}\label{assu:feas_MPCr}
	{Condition \eqref{eq:cond_prop1} holds true, implying the existence of feasible solutions of Problem \eqref{eq:MPC_mg_r}.} \eod
\end{assumption}

Therefore, the DMPC method presented in Algorithm \ref{alg:dual_dec} can be then applied to solve Problem \eqref{eq:MPC_mg_r} by simply substituting \eqref{eq:mg_loc_cons} with \eqref{eq:mg_loc_cons_r} in the local optimization problem, i.e., step 5.
\begin{remark}
	Problem \eqref{eq:MPC_mg_r} can also be expressed as a min-max problem \cite{bemporad1999}. However, in this robust reformulation \eqref{eq:MPC_mg_r}, the computational complexity is lower than that in the min-max counterpart. \eod
\end{remark}


%% file: at_det.tex
In this section, the methodology to identify the adversarial agents in the system and, at the same time, to block the attacks is presented.  It is an active detection strategy, where regular agents test their hypothesis to find their adversarial neighbors by deciding to open/close their connections with their neighbors. The methodology involves applying Bayesian inference for hypothesis testing (e.g., \cite{hoff2009}) and solving mixed-integer optimization problems. Note that in the control literature, Bayesian inference has also been applied to system identification \cite{ninnes2010} and fault detection \cite{fernandez2016}, while hypothesis testing has been used within the framework of fault diagnosis and robust control \cite{ocampo2018}.

Firstly, a regular agent, $i \in  \mathcal{R}$, detects an attack performed by one of its neighbors by evaluating its own SoC at the current time instant as follows:
\begin{equation}
	\Delta_{i,k} = \left|x_{i,k}-\left(x_{i,k-1}+\bm{b}_i^{\top}\bm{u}_{i,k-1}^{\star}+b_i\hat{p}^{\mathrm{d}}_{i,k-1}\right)\right|, \label{eq:det_at}
\end{equation}
where $\bm{b}_i = b_i [0 \ -\mathds{1}_{2+|\mathcal{N}_i|}^{\top}]^{\top}$. If $\Delta_{i,k} > b_id_{i}^{\mathrm{max}}$, then at $k$, agent $i$ is considered to be attacked, otherwise agent $i$ is not attacked.  
\begin{remark}
An attack $w^{\mathrm{a}}_{i,k}$ such that $|w^{\mathrm{a}}_{i,k}+w^{\mathrm{d}}_{i,k}|\leq  d_{i}^{\mathrm{max}}$ is undetectable since the regular agents cannot distinguish it from the load disturbance. However, such an attack is tolerable since the agents consider the bound of load disturbance as $d_{i}^{\mathrm{max}}$ in the first place.\eod  
\end{remark}

Although an attack can be detected, for $|\mathcal{N}_i| > 1$, it is not possible to determine which neighbor is the adversarial one by only evaluating \eqref{eq:det_at}. Therefore, 
in order to identify the adversarial neighbors, we apply a hypothesis testing method that is based on Bayesian inference \cite{hoff2009}. 

Each agent, $i \in \mathcal{R}$, considers the following set of hypotheses, $\mathcal{H}_i =\{\mathbf{H}^0_i,\mathbf{H}^j_i:j\in \mathcal{N}_i \}$, where the hypotheses are defined as follows:
\begin{itemize}
	\item $\mathbf{H}^0_i$ : There is no attack,
	\item $\mathbf{H}^j_i$ : Neighbor $j$ is an adversarial agent,
\end{itemize}
for all $j \in \mathcal{N}_i$. The Bayesian inference is used as the model to update the probability of the hypotheses as follows:
\begin{equation}
	\mathbb{P}_{k+1}(\mathbf{H}^j_i) = \frac{\mathbb{P}_k(\mathbf{H}^j_i)\mathbb{P}_k(\Delta_{i,k}|\mathbf{H}^j_i)}{\mathbb{P}_k(\Delta_{i,k})}, \label{eq:bay_inf}
\end{equation}
for all $\mathbf{H}^j_i \in \mathcal{H}_i$, where $\mathbb{P}_k(\mathbf{H}^j_i)$ denotes the probability of hypothesis $\mathbf{H}^j_i$ at time instant $k$, $\mathbb{P}_k(\Delta_{i,k})$ denotes the the marginal likelihood of $\Delta_{i,k}$, and $\mathbb{P}_k(\Delta_{i,k}|\mathbf{H}^j_i)$ denotes the probability of observing $\Delta_{i,k}$ given hypothesis $\mathbf{H}^j_i$  and is formulated as follows:
\begin{align*}
	\mathbb{P}_k(\Delta_{i,k}\leq b_id_{i}^{\mathrm{max}}|\mathbf{H}^j_i) &= \begin{cases}
	1, & \text{for $j = 0$,}\\
	1-v_{i,k}^jP_{\mathrm{at}}, & \text{for all $j \in \mathcal{N}_i$},
	\end{cases}\\
	\mathbb{P}_k(\Delta_{i,k} > b_id_{i}^{\mathrm{max}}|\mathbf{H}^j_i) &= \begin{cases}
	0, & \text{for $j = 0$},\\
	v_{i,k}^jP_{\mathrm{at}}, & \text{for all $j \in \mathcal{N}_i$},
	\end{cases}
\end{align*}
where $v_{i,k}^j \in  \{0,1\}$, for all $j \in \mathcal{N}_i$, denote the decision whether agent $i$ connects to and negotiates with neighbor $j$, i.e., $v_{i,k}^j=1$ implies agent $i$ connects to neighbor $j$, whereas $v_{i,k}^j=0$ implies agent $i$ does not connect to neighbor $j$. Note that $\mathbb{P}_{k+1}(\mathbf{H}^j_i)$ is the a posteriori probability of $\mathbf{H}^j_i$ given the event $\Delta_{i,k}$, i.e., $\mathbb{P}_{k+1}(\mathbf{H}^j_i)=\mathbb{P}(\mathbf{H}^j_i|\Delta_{i,k})$. The initial probabilities of all hypotheses are defined as
\begin{equation}
	\mathbb{P}_{0}(\mathbf{H}^j_i) = \begin{cases}
	1 - P_{\mathrm{at}}, &  \text{for $j = 0$,} \\
	P_{\mathrm{at}}/ |\mathcal{N}_i| & \text{for all $j \in \mathcal{N}_i$},
	\end{cases}
	\label{eq:in_PH}
\end{equation}
implying that it is initially considered that each neighbor is equally likely to be adversarial.

In order to decide the connection that a regular agent $i \in\mathcal{R}$ will have with its neighbors at each time instant, each agent $i \in \mathcal{R}$ solves a local mixed-integer optimization problem of the form:
\begin{subequations}
\begin{align}
\underset{\bm{v}_{i,k},\{\bm{u}_{i,\ell|k}\}_{\ell=k}^{k+h_p-1}}{\text{minimize}} & \sum_{\ell =k}^{k+h_p-1}  J_{i,\ell}(\bm{u}_{i,\ell|k}) + J^{\mathrm{v}}_{i}(\bm{v}_{i,k})\\
\text{subject to }  
& \bm{F}_i^{\mathrm{lc}}{\bm{u}}_{i,\ell|k} + \bm{F}_{\mathrm{v},i}^{\mathrm{lc}}\bm{v}_{i,k} \leq {\bm{f}}^{\mathrm{lc}}_{i,\ell}, \label{eq:mg_loc_cons1}\\
& \bm{v}_{i,k} \in \mathcal{C}_i \cup \{\mathds{1}_{|\mathcal{N}_i|} \}, \label{eq:v_set}
\end{align}
\label{eq:MPC_mg_local}%
\end{subequations}
where $\bm{v}_{i,k}=[v_{i,k}^j]^{\top}_{j \in \mathcal{N}_i}$. Here, the cost function $J^{\mathrm{v}}_{i}(\bm{v}_{i,k}):\mathbb{R}^{|\mathcal{N}_i|} \to \mathbb{R}$ penalizes the decision of having a connection with the neighbors. It is expressed as follows:
$$J^{\mathrm{v}}_{i}(\bm{v}_{i,k}) =  \gamma n_{\mathrm{at}} \sum_{j \in \mathcal{N}_i}\mathbb{P}_k(\mathbf{H}_i^j)(v_{i,k}^j)^{2},$$
where $\gamma \in \mathbb{R}_{>0}$ denotes a weight that can be tuned and $n_{\mathrm{at}}$ denotes the number of attacks that agent $i$ has received, i.e., the number of time instants at which $\Delta_{i,k} > b_i d_i^{\mathrm{max}}$. By having $n_{\mathrm{at}}$ as a weight, establishing a connection with a neighbor is penalized more if the number of received attacks increases. Moreover, \eqref{eq:mg_loc_cons1} is obtained from \eqref{eq:pow_bal1}, \eqref{eq:dyn_bat}, \eqref{eq:cap_pg}, \eqref{eq:cap_im}, \eqref{eq:cap_bat_r}, and \eqref{eq:ch_bat_r} as well as from the following expressions:
\begin{align}
w_{i,k}^{\mathrm{max}} = \max_{j \in \mathcal{N}_i}\left(2p^{\mathrm{t},\mathrm{max}}_{ji}v_{i,k}^{j}\right)+d_i^{\mathrm{max}}, \label{eq:w_max1}\\
-p^{\mathrm{t},\mathrm{max}}_{ji}v_{i,k}^j \leq p^{\mathrm{t}}_{ji,\ell} \leq p^{\mathrm{t},\mathrm{max}}_{ji}v_{i,k}^j, \quad \forall j \in\mathcal{N}_i, \label{eq:cap_t1}
\end{align}
for all $\ell \in \{k,\dots,k+h_p-1\}$, whereas, in the constraint \eqref{eq:v_set}, $\mathcal{C}_i =\{\bm{z}_j =\mathds{1}_{|\mathcal{N}_i|}-\bm{e}_j, j = 1,2,\dots,|\mathcal{N}_i| \}$, where $\bm{e}_j$, for all $j = 1,2,\dots|\mathcal{N}_i|$, are  the standard basis vectors of $|\mathcal{N}_i|$-dimensional Euclidean space. 

Problem \eqref{eq:MPC_mg_local} is a mixed-integer quadratic program (MIQP) due to the existence of $\bm{v}_{i,k}$. Notice that we penalize $v_{i,k}^j$, for each $j \in \mathcal{N}_i$, proportionally to the probability value of the hypothesis associated to neighbor $j$, $\mathbb{P}_k(\mathbf{H}_i^j)$. 
Furthermore, \eqref{eq:v_set} implies that agent $i$ only allows that it is disconnected from at most one neighbor. This means that there are only $|\mathcal{N}_i|+1$ possible solutions of $\bm{v}_{i,k}$. In addition, this constraint is added based on Assumption \ref{assu:f=1}.  
\begin{proposition}
	{Suppose that Assumptions \ref{as:feas_sol} and \ref{assu:feas_MPCr} hold. Then, Problem \eqref{eq:MPC_mg_local} has feasible solutions.} \eod
\end{proposition}
\begin{proof}
	 Any solution of $\bm{v}_{i,k} \in \mathcal{C}_i \cup \{\mathds{1}_{|\mathcal{N}_i|} \}$  implies the satisfaction of \eqref{eq:cond_prop1} since Assumption \ref{assu:feas_MPCr} holds and yields a feasible solution of  $\{\bm{u}_{i,\ell|k}\}_{\ell=k}^{k+h_p-1}$ by choosing $ p^{\mathrm{t}}_{ji,\ell} = 0$, for all $j \in\mathcal{N}_i$ and $\ell \in \{k,\dots,k+h_p-1\}$ since this solution satisfies \eqref{eq:cap_t1} and, according to Assumption \ref{as:feas_sol}, also satisfies \eqref{eq:pow_bal1},\eqref{eq:dyn_bat}-\eqref{eq:cap_im}.
\end{proof}

Finally, suppose that the decision $ \bm{v}_{i,k}^{\star}=[v_{i,k}^{j\star}]^{\top}_{j \in \mathcal{N}_i}$ is the solution obtained from solving Problem \eqref{eq:MPC_mg_local}. Now, instead of using \eqref{eq:w_max}, each agent $i \in \mathcal{R}$ computes the worst case of the disturbance by plugging $\bm{v}_{i,k}^{\star}$ into \eqref{eq:w_max1}. Thus, in the robust problem \eqref{eq:MPC_mg_r}, the local constraints \eqref{eq:mg_loc_cons_r} are switched by \eqref{eq:mg_loc_cons1} with $\bm{v}_{i,k} = \bm{v}_{i,k}^{\star},$ for all $i \in \mathcal{N}$. 

\subsection{Overall Scheme and Sub-optimality Bound} 
\label{sec:ov_sc}
The overall scheme of the proposed method is given in Algorithm \ref{alg:o_s}.
\begin{algorithm}
	\caption{\color{black}Resilient distributed algorithm, for $i \in \mathcal{R}$}
	\begin{algorithmic}[1]
		\State Initialize the hypothesis probabilities according to \eqref{eq:in_PH} 
		\For{$k = 1,2,\dots$}
		\State Evaluate \eqref{eq:det_at} to detect an attack
		\State Update the probability value of the hypotheses according to \eqref{eq:bay_inf}
		\If{$\mathbb{P}_k(\mathbf{H}_i^j)=1, \ j \in \mathcal{N}_i,$}
		\State ${v}_{i,k}^{j\star}=\begin{cases}
		0, \quad \text{for } \mathbb{P}_k(\mathbf{H}_i^j)=1,\\
		1, \quad \text{for } \mathbb{P}_k(\mathbf{H}_i^j)=0
		\end{cases}$ 
		\State \label{std1}  Compute $\bm{u}_{i,k|k}^{\star}$ by solving \eqref{eq:MPC_mg_r}, considering \eqref{eq:mg_loc_cons_r} is formed by \eqref{eq:pow_bal1}, \eqref{eq:dyn_bat}, \eqref{eq:cap_pg}, \eqref{eq:cap_im}, \eqref{eq:cap_bat_r}, \eqref{eq:ch_bat_r}, \eqref{eq:cap_t1} with $\bm{v}_{i,k} = \bm{v}_{i,k}^{\star}$, and $w_{i,k}^{\mathrm{max}}=d_i^{\mathrm{max}}$, using Algorithm \ref{alg:dual_dec}
		\Else
		\State Compute ${v}_{i,k}^{j\star}$, for all $j \in \mathcal{N}_i$, by solving \eqref{eq:MPC_mg_local}
		\State \label{std} Compute $\bm{u}_{i,k|k}^{\star}$ by solving \eqref{eq:MPC_mg_r}, considering \eqref{eq:mg_loc_cons_r} is formed by \eqref{eq:pow_bal1}, \eqref{eq:dyn_bat}, \eqref{eq:cap_pg}, \eqref{eq:cap_im}, \eqref{eq:cap_bat_r}, \eqref{eq:ch_bat_r}, \eqref{eq:w_max1} and \eqref{eq:cap_t1}, with $\bm{v}_{i,k} = \bm{v}_{i,k}^{\star}$, using Algorithm \ref{alg:dual_dec}
		\EndIf
		\State Apply $\bm{u}_{i,k|k}^{\star}$ and  $\bm{v}_{i,k}^{\star}$ 
		\EndFor
	\end{algorithmic}
	\label{alg:o_s}
\end{algorithm}
\begin{assumption}\label{assu:connection}
	{Any agent can temporarily disconnect the physical link between itself and its neighbors, respecting the decision of  $\bm{v}_{i,k}^{\star}$. Two agents, $i$ and $j$, where $(i,j) \in \mathcal{E}$, can only exchange energy if and only if ${v}_{i,k}^{j\star}={v}_{j,k}^{i\star}=1$.} \eod
\end{assumption}

Assumption \ref{assu:connection} implies that, although there exists a connection between agents $i$ and $j$, either of them can block the influence by closing the connection. The decisions obtained by performing Algorithm \ref{alg:o_s} are characterized by the following Proposition \ref{prop:out_dec}.
\begin{proposition}
	\label{prop:out_dec}
	{Suppose that Assumptions \ref{as:feas_sol}-\ref{assu:connection} hold. If the regular agents, i.e. all $i \in \mathcal{R}$, apply Algorithm \ref{alg:o_s}, then the obtained decision $\bm{u}_{i,k}^{\star}$, for all $i \in \mathcal{R}$, do not violate the operational constraints \eqref{eq:pow_bal1}-\eqref{eq:cap_t} under an attack that is defined by Definition \ref{def:attack}, for all $k \in \mathbb{Z}_{\geq 0}$.} \eod
\end{proposition} 

\begin{proof}
	A regular agent $i \in \mathcal{R}$  obtains its control inputs $\bm{u}_{i,k|k}^{\star}$ in either step 7 or 10, based on whether the adversarial neighbor has been identified or not. The difference between steps 7 and 10 is the definition of $w_{i,k}^{\mathrm{max}}$, and it is seen that $d_i^{\mathrm{max}}$ is smaller than or equal to $w_{i,k}^{\mathrm{max}}$ that is expressed in \eqref{eq:w_max1}. By Assumption \ref{assu:feas_MPCr}, $w_{i,k}^{\mathrm{max}}$, expressed in \eqref{eq:w_max1} or in step 7, satisfy \eqref{eq:cond_prop1}. In the case that $\bm{v}_{i,k}^{\star} = \mathds{1}_{|\mathcal{N}_i|}$, we obtain the original robustified problem \eqref{eq:MPC_mg_r} and the claim follows immediately from Proposition \ref{prop:1}. Now, we consider the case that one of the neighbor is blocked. Suppose that agent $j \in \mathcal{N}_i$ is blocked, i.e.,  $v_{i,k}^{j\star} =0$. The constraint \eqref{eq:cap_t1} yields the following equality constraint: $p^{\mathrm{t}}_{ji,\ell} = 0$, for all $\ell \in \{k,\dots,k+h_p-1\}$.  Assumptions \ref{as:feas_sol} and \ref{assu:feas_MPCr}  result in a feasible solution $\bm{u}_{i,k}^{\star}$, where $p^{\mathrm{t}\star}_{ji,k}=0$. Thus, the claim follows from Proposition \ref{prop:1}. Furthermore, by Assumption \ref{assu:connection}, agent $i$ is physically disconnected from agent $j$. Therefore, if agent $j$ is adversarial, then it cannot attack agent $i$. 
\end{proof}

\begin{remark}
	{The proposed attack identification and mitigation methods can be implemented along with any distributed optimization algorithm that is able to solve Problems \eqref{eq:MPC_mg} and \eqref{eq:MPC_mg_r}.} \eod
\end{remark}

We also provide a sub-optimality certificate of the control inputs obtained by performing Algorithm \ref{alg:o_s}, which is stated in Proposition \ref{prop:subopt}.
\begin{proposition}\label{prop:subopt}
{	Suppose that $\{\bm{u}_{i,\ell}^{o}\}_{\ell=k}^{k+h_p-1}$, for all $i \in \mathcal{N}$, are the minimizers of the following problem:}
	\begin{subequations}
		\begin{align}
		\underset{\{\{\bm{u}_{i,\ell|k}\}_{i \in \mathcal{N}}\}_{\ell=k}^{k+h_p-1}}{\text{minimize}} & \sum_{i \in \mathcal{N}} \sum_{\ell =k}^{k+h_p-1}  J_{i,\ell}(\bm{u}_{i,\ell|k})\\
		\text{subject to }  
		& \bm{F}_i{\bm{u}}_{i,\ell|k} \leq {\bm{f}}_{i,\ell}, \ \forall i \in \mathcal{N}, \label{eq:mg_loc_cons_l}
		\end{align}
		\label{eq:MPC_mg_l}%
	\end{subequations}
	{for all $\ell \in \{k,\dots,k+h_p-1\}$, $\{\bm{u}_{i,\ell}^{\star}\}_{\ell=k}^{k+h_p-1}$ denotes the solution obtained from Algorithm \ref{alg:o_s} in step 7 or 10, and  $\{\tilde{\bm{u}}_{i,\ell}^{*}\}_{\ell=k}^{k+h_p-1}$ denotes the solution obtained from solving Problem \eqref{eq:MPC_mg}. Then, the sub-optimality of the solution, i.e. $ \Delta J_{i,k} = \sum_{\ell =k}^{k+h_p-1}  \left(J_{i,\ell}(\bm{u}_{i,\ell}^{\star})-J_{i,\ell}(\tilde{\bm{u}}_{i,\ell}^{*})\right)$, is bounded as follows: 
		$$ \Delta J_{i,k} \leq \sum_{\ell =k}^{k+h_p-1}  \left(J_{i,\ell}(\bm{u}_{i,\ell}^{\star})-J_{i,\ell}(\bm{u}_{i,\ell}^{o})\right).$$}\eod
\end{proposition}
\begin{proof}
	The system achieves global optimal performance if all agents $i \in \mathcal{N}$ apply the solution obtained from solving \eqref{eq:MPC_mg}, implying the adversarial agents do not attack, and the forecast loads are equal to the actual ones.  
	We prove the proposition by showing that 
	the following inequalities hold: $$\sum_{\ell =k}^{k+h_p-1}J_{i,\ell}(\bm{u}_{i,\ell}^{\star}) \geq \sum_{\ell =k}^{k+h_p-1} J_{i,\ell}(\tilde{\bm{u}}_{i,\ell}^{*})\geq \sum_{\ell =k}^{k+h_p-1} J_{i,\ell}(\bm{u}_{i,\ell}^{o}).$$ 
	 Notice that Problem \eqref{eq:MPC_mg_l} is actually a relaxed formulation of Problem \eqref{eq:MPC_mg}, i.e., Problem \eqref{eq:MPC_mg} without constraint \eqref{eq:mg_coup_cons}. Therefore, any feasible solution of Problem \eqref{eq:MPC_mg} is also a feasible solution of Problem \eqref{eq:MPC_mg_l}, but not necessarily vice versa. Furthermore, the constraints imposed in the problem that is solved either in steps \ref{std1} or \ref{std} of Algorithm \ref{alg:o_s} are tighter than those in Problem \eqref{eq:MPC_mg}, implying any feasible solution obtained from applying Algorithm \ref{alg:o_s} is also feasible for Problem \eqref{eq:MPC_mg}, but not necessarily vice versa. 
\end{proof}
\begin{remark}
	{Problem \eqref{eq:MPC_mg_l} is trivially separable since there is no coupling constraint. Therefore, each agent $i \in \mathcal{N}$ can compute $\{\bm{u}_{i,k}^{o}\}_{\ell=k}^{k+h_p-1}$ independently as follows:} 
	\begin{align*}
		\{\bm{u}_{i,k}^{o}\}_{\ell=k}^{k+h_p-1} =  & \underset{\ \ \ \{\bm{u}_{i,\ell|k}\}_{\ell=k}^{k+h_p-1}}{\text{arg min}} \ \sum_{\ell =k}^{k+h_p-1}  J_{i,\ell}(\bm{u}_{i,\ell|k})\\ &\text{subject to} \   
		\bm{F}_i{\bm{u}}_{i,\ell|k} \leq {\bm{f}}_{i,\ell},
	\end{align*}
	 {for all} $\ell \in \{k,\dots,k+h_p-1\}$. \eod
\end{remark}


%% file: case_st.tex
As a case study, we use the PG\&E 69-bus distribution network, which has been modified by adding distributed generators and energy storages \cite{arefifar2012}, as depicted in Fig. \ref{fig:cs_top}. We follow  the partition given by \cite{arefifar2012} to divide the network into eight interconnected microgrids (agents).  The operational parameters of each microgrid are given in Table \ref{tab:param}. Furthermore, we consider two types of load profiles, which are industrial and residential, and assign each microgrid to one of the profiles randomly. Moreover, we generate the load profile and load forecast of each microgrid by considering the available load data as the maximum loads. In this case study, microgrids 2, 6, and 7 are chosen to be adversarial and the probability of attacks is set to be 0.3, which is known by the regular agents. Furthermore, the prediction horizon of each agent is $h_p = 4$ steps and we consider one-day simulation with sampling time of 15 minutes. 
\begin{figure}
	\centering
	\includegraphics[scale=1.0]{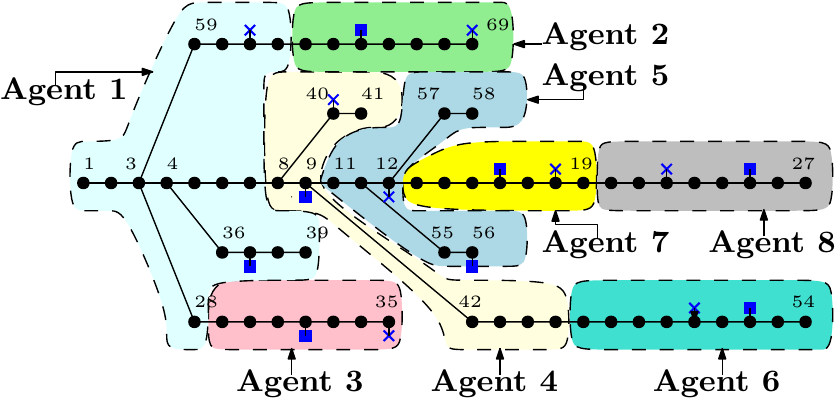}
	\caption{The topology of the PG\&E 69-bus distribution system and its 8-agent resulting partition. Blue crosses and squares indicate the distributed generators and storages, respectively.
	}
	\label{fig:cs_top}
\end{figure}

\begin{table}
	\centering
	\caption{Parameters of the Microgrids}
	\label{tab:param}
	\begin{tabular}{c c c c}
		\hline  \noalign{\smallskip}
		\textbf{Parameters} & \textbf{Value} & \textbf{Unit} & \textbf{Agent ($i$)}  \\ 
		\noalign{\smallskip}\hline  \noalign{\smallskip}
		$x^{\mathrm{min}}_i$, $x^{\mathrm{max}}_i$, $x_{i,0}$ & 40, 70, 55  & \% & all \\ 
		\noalign{\smallskip}
		$p^{\mathrm{ch}}_i$, $p^{\mathrm{dh}}_i$ & 300, 300 & kW & all \\ 
		\noalign{\smallskip}
		$p^{\mathrm{G,min}}_i$, $p^{\mathrm{G,max}}_i$  & 0,1500 & kW & all  \\ 		   
		\noalign{\smallskip}
		$p^{\mathrm{t,max}}_i$, $p^{\mathrm{im,max}}_i$ & 100, 2000 & kW & all  \\ 
		\noalign{\smallskip}
		$e_{\mathrm{cap},i}$ & 1000 & kWh & all  \\
		\noalign{\smallskip}
		$a_i$  & 1.0  & - & all  \\
		\noalign{\smallskip}
		$c^{\mathrm{st}}_i$, $c^{\mathrm{im}}_i$, $c^{\mathrm{t}}_i$ & 1, 250, 0.1 & - & all \\
		\noalign{\smallskip}
		$c^{\mathrm{g}}_i$ & 5 & - & 2, 3, 6, 7 \\		
		\noalign{\smallskip}
		$c^{\mathrm{g}}_i$ & 10 & - & 1, 4, 5, 8\\
		\noalign{\smallskip} \hline
	\end{tabular}
\end{table} 

\begin{algorithm}[H]
	\caption{Distributed robust algorithm, for $i \in \mathcal{R}$}
	\begin{algorithmic}[1]
		\For{$k = 1,2,\dots$}
		\State  \label{std2} Compute $\bm{u}_{i,k}$ by solving Problem \eqref{eq:MPC_mg_r}, considering \eqref{eq:mg_loc_cons_r} is formed by \eqref{eq:pow_bal1}, \eqref{eq:dyn_bat}, \eqref{eq:cap_pg}-\eqref{eq:cap_t}, and \eqref{eq:w_max}-\eqref{eq:ch_bat_r}, with Algorithm \ref{alg:dual_dec}
		\State Apply $\bm{u}_{i,k}$
		\EndFor
	\end{algorithmic}
	\label{alg:o_s_r}
\end{algorithm}

We consider four simulation scenarios, in each of which a different distributed strategy is applied (see Table \ref{tab:perf}). As the baseline performance, in Scenario 1, the nominal approach, i.e., applying Algorithm \ref{alg:dual_dec} to solve Problem \eqref{eq:MPC_mg}, is implemented for the case in which the adversarial agents do not attack and there is no load disturbance, whereas, in Scenario 2, the nominal approach is applied to the case with attacks and load disturbance. In Scenario 3, we apply the robustified approach without attack identification and mitigation as shown in  Algorithm \ref{alg:o_s_r}, while in Scenario 4,  we apply the proposed approach. 
Table \ref{tab:perf} shows the overall performance of the network over the whole simulation time. The proposed approach achieves a better performance than the robustified approach while ensuring the satisfaction of the constraints. As shown in Fig. \ref{fig:x}, in Scenario 2, the minimum limit of the SoC is violated. However, this violation does not occur in Scenarios 3 and 4. Moreover, Fig. \ref{fig:PhPb} shows how agent 1 detects agent 2 as the adversarial neighbor in Scenario 4. Once detected, i.e., at $k=8$, agent 1 disconnects from agent 2. Additionally, the average sub-optimality bound of the proposed approach is 49\% of the nominal performance (Scenario 1), whereas the measured sub-optimality is 18\%.    

\begin{table}
	\centering
	\caption{Total Cost of the System}
	\label{tab:perf}
	\begin{tabular}{c c c c c}
		\hline \noalign{\smallskip}
		\textbf{Scenario}&\textbf{Dist.}	& \textbf{Attack/Load} & \textbf{Cost (Pro-} & \textbf{{Constraint}}  \\ 
		&\textbf{Strategy} & \textbf{ Disturbance} &  \textbf{portional)}& \textbf{{Satisfaction}}\\
		\noalign{\smallskip} \hline
		\noalign{\smallskip} 
		1 & Nominal & No 	&  1.00 & Yes\\ 
		\noalign{\smallskip}
		2 & Nominal	& Yes & 1.06 & No\\ 
		\noalign{\smallskip} 
		3 & Alg. \ref{alg:o_s_r}	& Yes & 1.91 & Yes  \\ 
		\noalign{\smallskip}
		4  & Alg. \ref{alg:o_s} & Yes &  1.18 & Yes\\ 
		\noalign{\smallskip}
		\hline 
	\end{tabular} 
\end{table}

\begin{figure}
	\centering
	\includegraphics[scale=0.63]{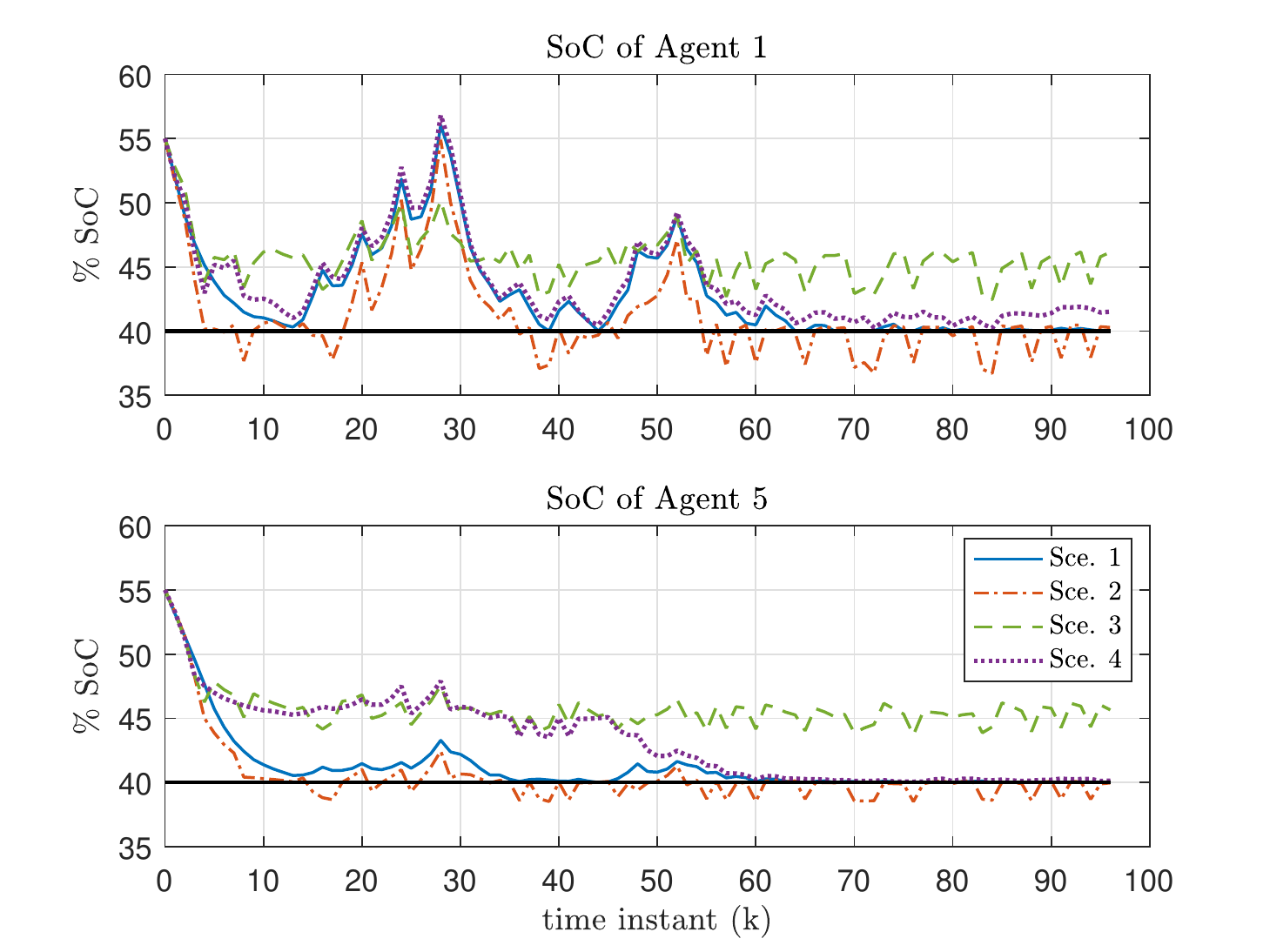}
	\caption{The evolution of the SoC of agents 1 (top) and 5 (bottom). The black horizontal line indicates the minimum limit of SoC, $x_i^{\mathrm{min}}$.
	}
	\label{fig:x}
\end{figure}
\vspace{-2pt}
\begin{figure}
	\centering
	\includegraphics[scale=0.63]{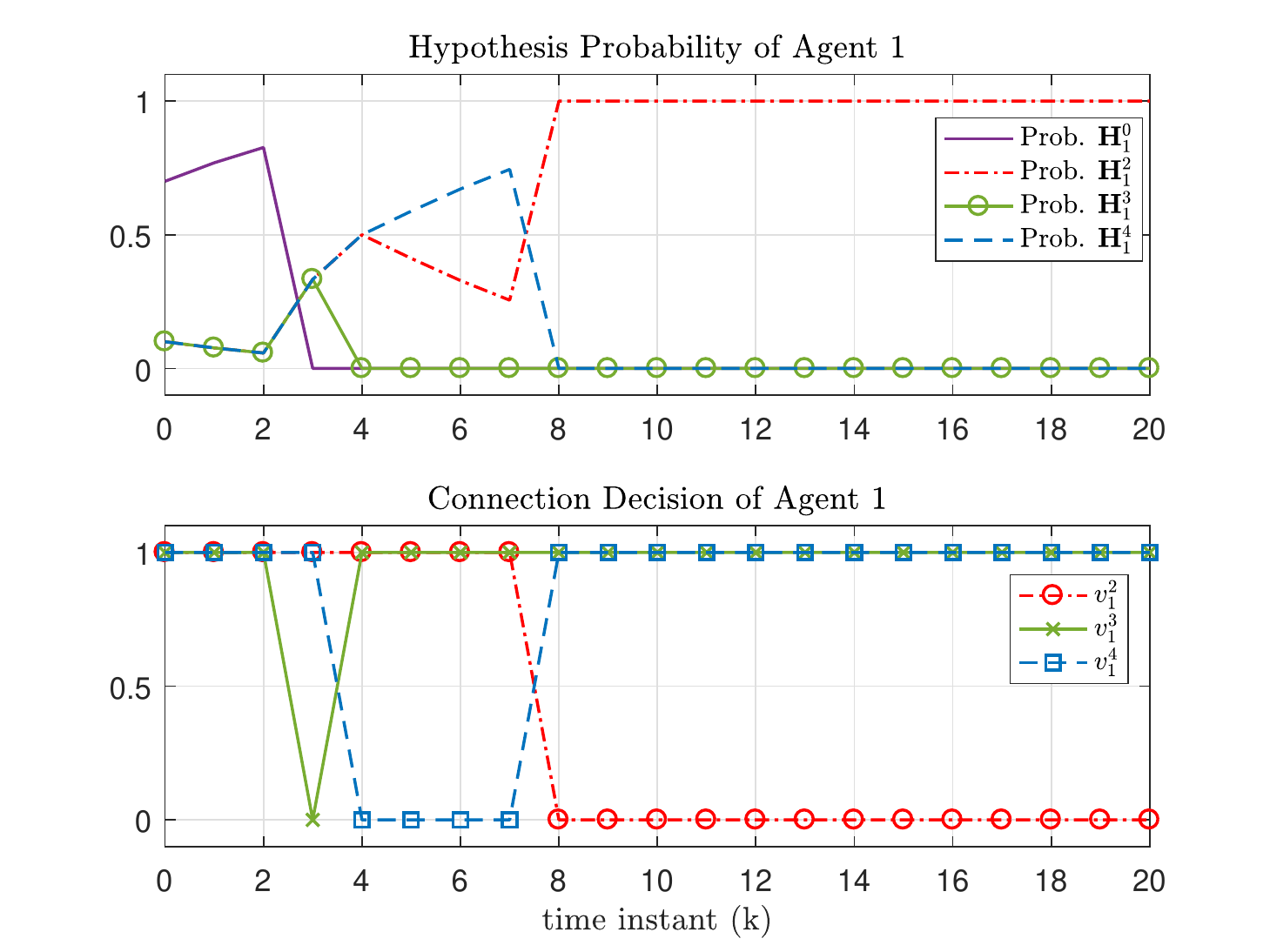}
	\caption{The evolution of the hypothesis probability (top) and the connection decision (bottom) of agent 1. Note that the decision $\bm{v}^{\star}_{1,k}$ are the same for $k = 8, 9, \dots, 96$, since the adversarial neighbor is detected at $k=8$.
	}
	\label{fig:PhPb}
\end{figure}

%% file: concl.tex
A distributed energy management for interconnected microgrid systems that is based on dynamic economic dispatch problem is investigated. We analyze the case of having microgrids that perform an adversarial behavior, i.e., some microgrids do not comply with the decisions obtained from the distributed strategy. Furthermore, we propose a robustified formulation and an attack identification and mitigation method such that the distributed strategy can deal with such adversaries. Additionally, we also provide a sub-optimality certificate of the proposed approach.

Future work includes extending the proposed approach such that the stochasticity of the loads is taken into account explicitly in order to improve the performance and assumptions on the number of adversarial neighbors are relaxed. Furthermore, we will also explore the possibility to improve the detection strategy as well as the attack mitigation method, e.g., by considering that the agents might exchange their hypothesis probability and by considering $\bm{v}_{i,k}$, for all $i\in \mathcal{N}$, as continuous variables that determine the limit of the connections among agents.